\documentclass[a4paper,english]{lipics-v2016}
 
\usepackage{microtype}
\usepackage[utf8]{inputenc}
\usepackage{placeins}

\usepackage{algorithmicx}
\usepackage[Algorithm,ruled]{algorithm}
\algrenewcommand\algorithmicindent{0.75em}
\algrenewcommand\alglinenumber[1]{\scriptsize #1:}
\pdfoutput=1
\newcommand{\oequiv}{\mathrel{\stackrel{\makebox[0pt]{\mbox{\normalfont\tiny o}}}{\equiv}}}
\newcommand{\function}{\textsc}
\newcommand{\tab}{\hspace{\algorithmicindent}}

\title{Non-uniform Replication}
\titlerunning{Non-uniform Replication}

\author[1]{Gonçalo Cabrita}
\author[2]{Nuno Preguiça}
\affil[1]{NOVA LINCS \& DI, FCT, Universidade NOVA de Lisboa, Caparica, Portugal\\
  \texttt{g.cabrita@campus.fct.unl.pt}}
\affil[2]{NOVA LINCS \& DI, FCT, Universidade NOVA de Lisboa, Caparica, Portugal\\
  \texttt{nuno.preguica@fct.unl.pt}}
\authorrunning{G. Cabrita and N. Preguiça}

\Copyright{Gonçalo Cabrita and Nuno Preguiça}

\subjclass{C.2.4 Distributed Systems} 
\keywords{Non-uniform Replication; Partial Replication; Replicated Data Types; Eventual Consistency}%

\begin{document}

\maketitle

\begin{abstract}
	Replication is a key technique in the design of efficient and reliable distributed systems. As information grows, it becomes difficult or even impossible to store all information at every replica. A common approach to deal with this problem is to rely on partial replication, where each replica maintains only a part of the total system information. As a consequence, a remote replica might need to be contacted for computing the reply to some given query, which leads to high latency costs particularly in geo-replicated settings. In this work, we introduce the concept of non-uniform replication, where each replica stores only part of the information, but where all replicas store enough information to answer every query. We apply this concept to eventual consistency and conflict-free replicated data types. We show that this model can address useful problems and present two data types that solve such problems. Our evaluation shows that non-uniform replication is more efficient than traditional replication, using less storage space and network bandwidth.
\end{abstract}

\section{Introduction}

Many applications run on cloud infrastructures composed by multiple data centers, 
geographically distributed across the world. These applications usually store their data on 
geo-replicated data stores, with replicas of data being maintained in multiple data centers.
Data management in geo-replicated settings is challenging, requiring designers to make
a number of choices to better address the requirements of applications.

One well-known trade-off is between availability and data consistency.
Some data stores provide strong consistency \cite{spanner,replicatedcommit}, where the system gives 
the illusion that a single replica exists. 
This requires replicas to coordinate for executing operations, with impact on the latency and availability 
of these systems. 
Other data stores~\cite{dynamo,cassandra} provide high-availability and low latency by allowing operations 
to execute locally in a single data center eschewing a linearizable consistency model. These systems receive and
execute updates in a single replica before asynchronously propagating the updates to other replicas, 
thus providing very low latency. 

With the increase of the number of data centers available to applications and the amount of
information maintained by applications, another trade-off is between the simplicity of
maintaining all data in all data centers and the cost of doing so.
Besides sharding data among multiple machines in each data center, it is often interesting 
to keep only part of the data in each data center to reduce the costs associated with data storage
and running protocols that involve a large number of replicas. 
In systems that adopt a partial replication model~\cite{pstore,Terry13Consistency,Tyler15Designing}, 
as each replica only maintains part of the data, it can only locally process a subset of the database queries. 
Thus, when executing a query in a data center, it might be necessary to contact one or more remote 
data centers for computing the result of the query. 

In this paper we explore an alternative partial replication model, the non-uniform replication model, 
where each replica maintains only part of the data but can process all queries. The key insight is 
that for some data objects, not all data is necessary for providing the result of read operations. 
For example, an object that keeps the top-K elements only needs to maintain those top-K elements 
in every replica. However, the remaining elements are necessary if a remove operation is available,
as one of the elements not in the top needs to be promoted when a top element is removed. 

A top-K object could be used for maintaining the leaderboard in an online game.
In such system, while the information for each user only needs to be kept in the data center 
closest to the user (and in one or two more for fault tolerance), it is important to keep a replica
of the leaderboard in every data center for low latency and availability. 
Currently, for supporting such a feature, several designs could be adopted. 
First, the system could maintain an object with the results of all players in all replicas.
While simple, this approach turns out to be needlessly expensive in both storage space and network bandwidth when compared to our proposed model.
Second, the system could move all data to a single data center and execute the computation 
in that data center or use a data processing system that can execute computations over geo-partitioned 
data \cite{pixida}. The result would then have to be sent to all data centers. 
This approach is much more complex than our proposal, and while it might be interesting when complex
machine learning computations are executed, it seems to be an overkill in a number of situations.

We apply the non-uniform replication model to eventual consistency and 
Conflict-free Replicated Data Types~\cite{crdts}, formalizing 
the model for an operation-based replication approach. We present two useful data type designs that
implement such model. 
Our evaluation shows that the non-uniform replication model leads to high gains in both storage space and 
network bandwidth used for synchronization when compared with state-of-the-art replication based 
alternatives.

In summary, this paper makes the following contributions:

\begin{itemize}
	
	\item The proposal of the non-uniform replication model, where each replica only keeps
	part of the data but enough data to reply to every query;
	
	\item The definition of non-uniform eventual consistency (NuEC), the identification
	of sufficient conditions for providing NuEC and a protocol that enforces such conditions
	relying on operation-based synchronization; 
	
	\item Two useful replicated data type designs that adopt the
	non-uniform replication model (and can be generalized to use different filter functions);
	
	\item An evaluation of the proposed model, showing its gains in term of storage
	space and network bandwidth.
	
\end{itemize}

The remainder of this paper is organized as follows. Section \ref{sec:relatedwork} discusses the related work. Section \ref{sec:nu-replicated-systems} describes the non-uniform replication model. Section \ref{sec:nu-eventual-consistency} applies the model to an eventual consistent system. Section \ref{sec:opbased} introduces two useful data type designs that follow the model. Section \ref{sec:evaluations} compares our proposed data types against state-of-the-art CRDTs.

\section{Related Work}
\label{sec:relatedwork}

\textbf{Replication:}
A large number of replication protocols have been proposed in the last decades~\cite{epidemic,eventual,cops,eiger,chainreaction,saito,replicatedcommit}.
Regarding the contents of the replicas, these protocols can be divided in those
providing full replication, where each replica maintains the full database state, 
and partial replication, where each replica maintains only a subset of the database
state.

Full replication strategies allow operations to concurrently modify all replicas 
of a system and, assuming that replicas are mutually consistent, improves 
availability since clients may query any replica in the system and obtain an 
immediate response. While this improves the performance of read operations, 
update operations now negatively affect the performance of the system since 
they must modify every replica which severely affects middle-scale to 
large-scale systems in geo-distributed settings. 
This model also has the disadvantage of limiting the system's total 
capacity to the capacity of the node with fewest resources.

Partial replication~\cite{alonso1997partial,pstore,Terry13Consistency,Tyler15Designing} addresses the shortcomings of full replication by having each replica store only part of the data (which continues being replicated in more than one node). This improves the scalability of the system but since each replica maintains only a part of the data, it can only locally process a subset of queries. This adds complexity to the query processing, with
some queries requiring contacting multiple replicas to compute their result.
In our work we address these limitations by proposing a model where each 
replica maintains only part of the data but can reply to any query.

Despite of adopting full or partial replication, replication protocols enforce strong consistency~\cite{replicatedcommit,spanner,blotter}, weak consistency~\cite{eventual,dynamo,cops,eiger,chainreaction}
or a mix of these consistency models~\cite{psi,redblue}. 
In this paper we show how to combine non-uniform replication with eventual consistency. 
An important aspect in systems that adopt eventual consistency is how the system handles 
concurrent operations.
CRDTs have been proposed as a technique for addressing such challenge.

\textbf{CRDTs:}
Conflict-free Replicated Data Types~\cite{crdts} are data types designed to be
replicated at multiple replicas without requiring coordination for executing operations.
CRDTs encode merge policies used to guarantee that all replicas converge to the same value 
after all updates are propagated to every replica.
This allows an operation to execute immediately on any replica, with replicas synchronizing
asynchronously. Thus, a system that uses CRDTs can provide low latency and high availability,
despite faults and network latency.
With these guarantees, CRDTs are a key building block for providing eventual consistency 
with well defined semantics, making it easier for programmers to reason about the system
evolution.

When considering the synchronization process, two main types of CRDTs have been proposed:
state-based CRDT, where replicas synchronize pairwise, by periodically exchanging the state of 
the replicas; and operation-based CRDTs, where all operations need to be propagated to 
all replicas. 

Delta-based CRDTs~\cite{delta} improve upon state-based CRDTs by 
reducing the dissemination cost of updates, sending only a delta of the modified state. 
This is achieved by using \textsl{delta-mutators}, which are functions that 
encode a delta of the state. 
Linde et. al~\cite{bigdelta} propose an improvement to delta-based CRDTs that further
reduce the data that need to be propagated when a replica first synchronizes with some
other replica. This is particularly interesting in peer-to-peer settings, where the 
synchronization partners of each replica change frequently.
Although delta-based CRDTs reduce the network bandwidth used for synchronization, 
they continue to maintain a full replication strategy where the state of 
quiescent replicas is equivalent.

Computational CRDTs~\cite{Navalho:2015:SCC:2745947.2745948} are an extension of
state-based CRDTs where the state of the object is the result of a
computation (e.g. the average, the top-K elements) over the executed updates.
As with the model we propose in this paper, replicas do not need to
have equivalent states.
The work we present in this paper extends the initial ideas proposed in
computational CRDTs in several aspects, including the definition of the
non-uniform replication model, its application to operation-based eventual
consistency and the new data type designs.

\section{Non-uniform replication}
\label{sec:nu-replicated-systems}

We consider an asynchronous distributed system composed by $n$ nodes.
Without loss of generality, we assume that the system replicates a single
object. 
The object has an interface composed by a set of read-only operations,
$\cal Q$, and a set of update operations, $\cal U$.
Let $\cal S$ be the set of all possible object states, the state that 
results from executing operation $o$ in
state $s \in \cal S$ is denoted as $s \bullet o$.
For a read-only operation, $q \in {\cal Q}$, $s \bullet q = s$.
The result of operation $o \in {\cal Q} \cup {\cal U}$ in
state $s \in {\cal S}$ is denoted as $o(s)$ (we assume that
an update operation, besides modifying the state, can also
return some result). 

We denote the state of the replicated system as a tuple
$(s_1,s_2,\ldots,s_n)$, with $s_i$ the state of the replica $i$.
The state of the replicas is synchronized by a replication protocol
that exchanges messages among the nodes of the system and updates
the state of the replicas.
For now, we do not consider any specific replication protocol or strategy, 
as our proposal can be applied to different replication strategies.

We say a system is in a quiescent state for a given set of executed operations if
the replication protocol has propagated all messages necessary to synchronize 
all replicas, i.e., additional messages sent by the replication protocol will
not modify the state of the replicas.
In general, replication protocols try to achieve a convergence property,
in which the state of any two replicas is equivalent in a quiescent state.

\begin{definition}[Equivalent state]
	Two states, $s_i$ and $s_j$, are \emph{equivalent}, $s_i \equiv s_j$, iff
	the results of the execution of any sequence of operations in both states are equal, i.e.,
	$\forall o_1, \ldots, o_n \in {\cal Q} \cup {\cal U}, 
	o_n(s_i \bullet o_1 \bullet \ldots \bullet o_{n-1}) = o_n(s_j \bullet o_1 \bullet \ldots \bullet o_{n-1})$.
\end{definition}

This property is enforced by most replication protocols, independently of whether they 
provide strong or weak consistency \cite{paxos,cops,eventual}.
We note that this property does not require that the internal state of the replicas is the same, 
but only that the replicas always return the same results for
any executed sequence of operations.

In this work, we propose to relax this property by requiring only 
that the execution of read-only operations return the same value. We name
this property as \emph{observable equivalence} and define it formally as follows.

\begin{definition}[Observable equivalent state]
	Two states, $s_i$ and $s_j$, are \emph{observable equivalent}, $s_i \oequiv s_j$, iff
	the result of executing every read-only operation 
	in both states is equal, i.e., $\forall o \in {\cal Q}, o(s_i) = o(s_j)$.
\end{definition}

As read-only operations do not affect the state of a replica, the results of the
execution of any sequence of read-only operations in two observable equivalent states will 
also be the same.
We now define a non-uniform replication system as one that guarantees only
that replicas converge to an observable equivalent state.

\begin{definition}[Non-uniform replicated system]
	We say that a replicated system is non-uniform if the replication 
	protocol guarantees that in a quiescent state, the state of any two replicas 
	is observable equivalent, i.e., in the quiescent state $(s_1,\ldots,s_n)$, we 
	have $s_i \oequiv s_j, \forall s_i,s_j \in \{s_1,\ldots,s_n\}$.  
\end{definition}

\subsection{Example}\label{sec:nurep:example}

We now give an example that shows the benefit of non-uniform replication.
Consider an object \emph{top-1} with three operations: 
(i) $add(name,value)$, an update operation that adds the pair to the top;
(ii) $rmv(name)$, an update operation that removes all previously added 
	pairs for \emph{name};
(iii) $get()$, a query that returns the pair
	with the largest value (when more than one pair has the same largest value, the one 
	with the smallest lexicographic name is returned).

Consider that $add(a,100)$ is executed in a replica and replicated to all replicas.
Later $add(b,110)$ is executed and replicated. At this moment, all
replicas know both pairs.

If later $add(c,105)$ executes in some replica,
the replication protocol does not need to propagate the update to the other
replicas in a non-uniform replicated system. 
In this case, all replicas are observable equivalent, as a query executed at any replica 
returns the same correct value.
This can have an important impact not only in the size of object replicas, 
as each replica will store only part of the data, but also in the bandwidth used
by the replication protocol, as not all updates need to be propagated to all replicas.

We note that the states that result from the previous execution are not equivalent 
because after executing $rmv(b)$, the get 
operation will return $(c,105)$ in the replica that has received the $add(c,105)$ we operation and
$(b,100)$ in the other replicas. 

Our definition only forces the states to be observable equivalent after the replication
protocol becomes quiescent. 
Different protocols can be devised giving different guarantees.
For example, for providing linearizability, the protocol should guarantee that
all replicas return $(c,105)$ after the remove.
This can be achieved, for example, by replicating the now relevant $(c,105)$ update in
the process of executing the remove. 

In the remainder of this paper, we study how to apply the concept of non-uniform
replication in the context of eventually consistent systems.
The study of its application to systems that provide strong consistency is left 
for future work.

\section{Non-uniform eventual consistency}
\label{sec:nu-eventual-consistency}

We now apply the concept of non-uniform replication to replicated systems
providing eventual consistency.

\subsection{System model}

We consider an asynchronous distributed system composed by $n$ nodes, where nodes may 
exhibit fail-stop faults but not byzantine faults. 
We assume a communication system with a single communication 
primitive, $mcast(m)$, that can be used by a process to send a message
to every other process in the system with reliable broadcast semantics. 
A message sent by a correct process is eventually received by all correct processes.
A message sent by a faulty process is either received by all correct processes
or none. 
Several communication systems provide such properties -- e.g. systems that  
propagate messages reliably using anti-entropy protocols~\cite{epidemic, epidemic2}.

An object is defined as a tuple $({\cal S},s^0,{\cal Q},{\cal U}_p,{\cal U}_e)$, 
where ${\cal S}$ is the set of valid 
states of the object, $s^0 \in {\cal S}$ is the initial state of the object, 
${\cal Q}$ is the set of read-only operations (or \emph{queries}), 
${\cal U}_p$ is the set of prepare-update operations and 
${\cal U}_e$ is the set of effect-update operations.

A query executes only at the replica where the operation is invoked, its source, 
and it has no side-effects, i.e., the state of an object remains unchanged after
executing the operation.
When an application wants to update the state of the object, it issues a
prepare-update operation, $u_p \in {\cal U}_p$. 
A $u_p$ operation executes only at the source, has no side-effects and
generates an effect-update operation, $u_e \in {\cal U}_e$.
At source, $u_e$ executes immediately after $u_p$.

As only effect-update operations may change the state of the object,
for reasoning about the evolution of replicas we can restrict our 
analysis to these operations.
To be precise, the execution of a prepare-update operation generates an instance 
of an effect-update operation. 
For simplicity, we refer the instances of operations simply as operations.
With $O_i$ the set of operations generated at node $i$, the set of operations
generated in an execution, or simply the set of operations in an execution,
is $O = O_1 \cup \ldots \cup O_n$.

\subsection{Non-uniform eventual consistency}

For any given execution, with $O$ the operations of the execution,
we say a replicated system provides \emph{eventual consistency} iff in a quiescent
state: 
(i) every replica executed all operations of $O$; and
(ii) the state of any pair of replicas is equivalent.

A sufficient condition for achieving the first property is to 
propagate all generated operations using reliable broadcast
(and execute any received operation).
A sufficient condition for achieving the second property is to
have only commutative operations. 
Thus, if all operations commute with each other, the execution 
of any serialization of $O$ in the initial state of the object 
leads to an equivalent state.

From now on, unless stated otherwise, we assume that all operations
commute.
In this case, as all serializations of $O$ are equivalent, we denote 
the execution of a serialization of $O$ in state $s$ simply as $s \bullet O$.

For any given execution, with $O$ the operations of the execution,
we say a replicated system provides \emph{non-uniform eventual consistency} 
iff in a quiescent state the state of any replica is observable equivalent to the 
state obtained by executing some serialization of $O$.
As a consequence, the state of any pair of replicas is also observable equivalent.

For a given set of operations in an execution $O$, we say that 
$O_{core} \subseteq O$ is a set of core operations of $O$ iff
$s^0 \bullet O \oequiv s^0 \bullet O_{core}$.
We define the set of operations that are irrelevant to the final state of the
replicas as follows: $O_{masked} \subseteq O$ is a set of masked operations of $O$ iff
$s^0 \bullet O \oequiv s^0 \bullet (O \setminus O_{masked})$.

\begin{theorem}[Sufficient conditions for NuEC]
	A replication system provides \emph{non-uniform eventual consistency (NuEC)}
	if, for a given set of operations $O$, the following conditions hold:
	(i) every replica executes a set of core operations of $O$; and 
	(ii) all operations commute.
\end{theorem}
\begin{proof}
	From the definition of core operations of $O$, and by the fact that all operations
	commute, it follows immediately that if a replica executes a set of core operations,
	then the final state of the replica is observable equivalent to the state obtained by
	executing a serialization of $O$.
	Additionally, any replica reaches an observable equivalent state.
\end{proof}

\subsection{Protocol for non-uniform eventual consistency}

We now build on the sufficient conditions for providing \emph{non-uniform 
	eventual consistency} to devise a correct replication protocol that tries
to minimize the operations propagated to other replicas.
The key idea is to avoid propagating operations that are part of 
a masked set. 
The challenge is to achieve this by using only local information, which 
includes only a subset of the executed operations.

\begin{algorithm}[ht]
	\caption{Replication algorithm for non-uniform eventual consistency}\label{replication-alg}
	\footnotesize
	\begin{algorithmic}[1]
		\State $S$ : state: initial $s^0$ \Comment{Object state}
		\State $log_{recv}$ : set of operations: initial $\{\}$
		\State $log_{local}$ : set of operations: initial $\{\}$ \Comment{Local operations not propagated}
		\\
		\State \function{execOp}($op$): void \Comment{New operation generated locally} \label{m:execOp}
		\State \tab $log_{local} = log_{local} \cup \{op\}$
		\State \tab $S = S \bullet op$
		\\
		\State \function{opsToPropagate}(): set of operations \Comment{Computes the local operations that need to be propagated}
		\State \tab $ops = maskedForever( log_{local}, S, log_{recv})$
		\State \tab $log_{local} = log_{local} \setminus ops$
		\State \tab $opsImpact = hasObservableImpact( log_{local}, S, log_{recv})$\label{m:opsimapct}
		\State \tab $opsPotImpact = mayHaveObservableImpact( log_{local}, S, log_{recv})$\label{m:opspotimapct}
		\State \tab \textbf{return} $opsImpact \cup opsPotImpact$
		\\
		\State \function{sync}(): void \Comment{Propagates local operations to remote replicas}
		\State \tab $ops = opsToPropagate()$
		\State \tab $compactedOps = compact(ops)$ \Comment{Compacts the set of operations}
		\State \tab $mcast(compactedOps)$
		\State \tab $log_{coreLocal} = \{\}$
		\State \tab $log_{local} = log_{local} \setminus ops$
		\State \tab $log_{recv} = log_{recv} \cup ops$
		\\
		\State \function{on receive}($ops$): void \Comment{Process remote operations}\label{m:onreceive}
		\State \tab $log_{recv} = log_{recv} \cup ops$
		\State \tab $S = S \bullet ops$
	\end{algorithmic}
\end{algorithm}

Algorithm~\ref{replication-alg} presents the pseudo-code of an algorithm 
for achieving \emph{non-uniform eventual consistency} -- the algorithm does
not address the durability of operations, which will be discussed later.

The algorithm maintains the state of the object and two sets of operations:
$log_{local}$, the set of effect-update operations generated in the local
replica and not yet propagated to other replicas;
$log_{recv}$, the set of effect-update operations propagated to all replicas
(including operations generated locally and remotely).

When an effect-update operation is generated, the \emph{execOp} function 
is called. This function adds the new operation to the log of local
operations and updates the local object state. 

The function \emph{sync} is called to propagate local operations to remote
replicas. It starts by computing which new operations need to be propagated, 
compacts the resulting set of operations for efficiency purposes, 
multicasts the compacted set of operations, and finally updates the local sets of operations.
When a replica receives a set of operations (line~\ref{m:onreceive}), the
set of operations propagated to all nodes and the local object state are updated accordingly.

Function \emph{opsToPropagate} addresses the key challenge of deciding which
operations need to be propagated to other replicas.
To this end, we divide the operations in four groups.

First, the \emph{forever masked} operations, which are operations that will
remain in the set of masked operations independently of the operations that might
be executed in the future. 
In the top example, an operation that adds a pair masks forever all known operations
that added a pair for the same element with a lower value.
These operations are removed from the set of local operations.

Second, the \emph{core} operations ($opsImpact$, line \ref{m:opsimapct}), as computed locally.
These operations need to be propagated, as they will (typically) impact
the observable state at every replica.

Third, the operations that might impact the observable state when considered
in combination with other non-core operations that might have been executed in other 
replicas ($opsPotImpact$, line \ref{m:opspotimapct}).
As there is no way to know which non-core operations have been executed in other replicas, 
it is necessary to propagate these operations also.
For example, consider a modified top object where the value associated with each element is the sum 
of the values of the pairs added to the object. In this case, 
an add operation that would not move an element to the top in a replica would be in this 
category because it could influence the top when combined with other concurrent adds for the
same element.

Fourth, the remaining operations that might impact the observable state in 
the future, depending on the evolution of the observable state.
These operations remain in $log_{local}$.
In the original top example, an operation that adds a pair that will not be in 
the top, as computed locally, is in this category as it might become the top element
after removing the elements with larger values.

For proving that the algorithm can be used to provide non-uniform eventual 
consistency, we need to prove the following property.

\begin{theorem}
	\label{theorem-algorithm}
	Algorithm~\ref{replication-alg} guarantees that in a quiescent state, considering
	all operations $O$ in an execution, all replicas have received all operations in
	a core set $O_{core}$. 
\end{theorem}
\begin{proof}
	To prove this property, we need to prove that there exists no operation that
	has not been propagated by some replica and that is required for any $O_{core}$
	set.
	Operations in the first category have been identified as masked operations
	independently of any other operations that might have been or will be executed.
	Thus, by definition of masked operations, a $O_{core}$ set will not (need to) include
	these operations.
	The fourth category includes operations that do not influence the observable 
	state when considering all executed operations -- if they might have impact,
	they would be in the third category. Thus, these operations do not need 
	to be in a $O_{core}$ set.
	All other operations are propagated to all replicas. Thus, in a quiescent state,
	every replica has received all operations that impact the observable state.
\end{proof}

\subsection{Fault-tolerance}\label{sec:nuec:ft}

Non-uniform replication aims at reducing the cost of communication 
and the size of replicas, by avoiding propagating operations that do
not influence the observable state of the object. 
This raises the question of the durability of operations that are not 
immediately propagated to all replicas.

One way to solve this problem is to have the source replica propagating
every local operation to $f$ more replicas to tolerate $f$ faults. 
This ensures that an operation survives even in the case of $f$ faults.
We note that it would be necessary to adapt the proposed algorithm, so that 
in the case a replica receives an operation for durability reasons, it 
would propagate the operation to other replicas if the source replica fails.
This can be achieved by considering it as any local operation (and introducing 
a mechanism to filter duplicate reception of operations).

\subsection{Causal consistency}

Causal consistency is a popular consistency model for replicated 
systems \cite{cops,chainreaction,eiger}, in which a replica only executes 
an operation after executing all operations that causally precede it~\cite{happensbefore}.
In the non-uniform replication model, it is impossible to strictly adhere to 
this definition because some operations are not propagated (immediately), which would
prevent all later operations from executing.
 
An alternative would be to restrict the dependencies to the execution of core
operations. The problem with this is that the status of an operation may change
by the execution of another operation. When a non-core operation becomes core, 
a number of dependencies that should have been enforced might have been missed
in some replicas.

We argue that the main interest of causal consistency, when compared 
with eventual consistency, lies in the semantics provided by the object.
Thus, in the designs that we present in the next section, we aim to guarantee that
in a quiescent state, the state of the replicated objects
provide equivalent semantics to that of a system that enforces causal consistency.

\section{Non-uniform operation-based CRDTs}
\label{sec:opbased}

CRDTs \cite{crdts} are data-types that can be replicated, modified
concurrently without coordination and guarantee the eventual consistency of
replicas given that all updates propagate to all replicas.
We now present the design of two useful operation-based CRDTs \cite{crdts} 
that adopt the non-uniform replication model. 
Unlike most operation-based CRDT designs, we do not assume that the system propagates
operations in a causal order. 
These designs were inspired by the state-based computational CRDTs proposed by 
Navalho \textit{et al.} \cite{Navalho:2015:SCC:2745947.2745948}, which also allow
replicas to diverge in their quiescent state.

\subsection{Top-K with removals NuCRDT}
\label{sec:topkrmv}

In this section we present the design of a non-uniform top-K CRDT, as the one
introduced in section~\ref{sec:nurep:example}. The data type allows access to the top-K elements added and can be used, for example, for 
maintaining the leaderboard in online games. 
The proposed design could be adapted to define any CRDT that filters elements
based on a deterministic function by replacing the \emph{topK} function used in the algorithm by another filter function.  

For defining the semantics of our data type, we start by defining the happens-before
relation among operations.
To this end, we start by considering the happens-before relation established 
among the events in the execution of the replicated system \cite{happensbefore}. 
The events that are considered relevant are: the generation of an operation at the source replica,
and the dispatch and reception of a message with a new operation or information that
no new message exists.
We say that operation $op_i$ happens before operation $op_j$ iff 
the generation of $op_i$ happened before the generation of $op_j$ in the
partial order of events.

The semantics of the operations defined in the top-K CRDT is the following.
The \emph{add(el,val)} operation adds a new pair to the object.
The \emph{rmv(el)} operation removes any pair of $el$ that was
added by an operation that happened-before
the \emph{rmv} (note that this includes non-core add operations that have not been propagated to the 
source replica of the remove).
This leads to an \emph{add-wins} policy \cite{crdts}, where a remove has no impact on concurrent
adds. 
The \emph{get()} operation returns the top-K pairs in the object, as defined by the function 
\emph{topK} used in the algorithm. 

\begin{algorithm}[ht]
	\caption{Top-K NuCRDT with removals}\label{topk-with-deletes}
	\footnotesize	
	\begin{algorithmic}[1]
		\State $\emph{elems}$ : set of $\langle \emph{id}, \emph{score}, \emph{ts} \rangle$ : initial $\{\}$
		\State $\emph{removes}$ : map $\emph{id} \mapsto \emph{vectorClock}$: initial $[]$
		\State $\emph{vc}$ : $\emph{vectorClock}$: initial $[]$
		\\
		\State \function{get}$()$ : set
		\State \tab \textbf{return} $\{\langle \emph{id}, \emph{score} \rangle : \langle \emph{id}, \emph{score}, \emph{ts} \rangle \in \emph{topK(elems)}\}$
		\\
		\State \textbf{prepare} \function{add}$(\emph{id}, \emph{score})$
		\State \tab \textbf{generate} $\emph{add(id}, \emph{score}, \langle\emph{getReplicaId()},++\emph{vc}[\emph{getReplicaId()}]\rangle\emph{)}$
		\\
		\State \textbf{effect} \function{add}$(\emph{id}, \emph{score}, \emph{ts})$
		\State \tab \textbf{if} $\emph{removes}[\emph{id}][\emph{ts.siteId}] < \emph{ts.val}$ \textbf{then}\label{alg:topk:wasdeleted}
		\State \tab \tab $\emph{elems} = \emph{elems} \cup \{\langle \emph{id}, \emph{score}, \emph{ts} \rangle\}$
		\State \tab \tab $\emph{vc}[\emph{ts.siteId}] = \emph{max(vc}[\emph{ts.siteId}], \emph{ts.val)}$
		\\
		\State \textbf{prepare} \function{rmv}$(\emph{id})$
		\State \tab \textbf{generate} $\emph{rmv(id,vc)}$
		\\
		\State \textbf{effect} \function{rmv}$(\emph{id}, \emph{vc}_\emph{{rmv}})$
		\State \tab $\emph{removes}[\emph{id}] = \emph{pointwiseMax}(\emph{removes}[\emph{id}], \emph{vc}_\emph{{rmv}})$
		\State \tab $\emph{elems} = \emph{elems} \setminus \{\langle \emph{id}_0, \emph{score}, \emph{ts} \rangle \in \emph{elem}: \emph{id} = \emph{id}_0 \wedge \emph{ts.val} \leq \emph{vc}_\emph{{rmv}}[\emph{ts.siteId}] \}$ \label{alg:topk:todelete}
		\\
		\State \function{maskedForever}$(\emph{log}_\emph{{local}}, S, \emph{log}_\emph{{recv}})$: set of operations
		\State \tab $\emph{adds} = \{\emph{add(id}_1, \emph{score}_1, \emph{ts}_1\emph{)} \in \emph{log}_\emph{{local}} :$
		\State \tab\tab $(\exists \emph{add(id}_2, \emph{score}_2, \emph{ts}_2\emph{)} \in \emph{log}_\emph{{local}} : \emph{id}_1 = \emph{id}_2 \land \emph{score}_1 < \emph{score}_2 \land \emph{ts}_1.\emph{val} < \emph{ts}_2.\emph{val}) \vee $
		\State \tab\tab\tab $(\exists \emph{rmv(id}_3, \emph{vc}_\emph{{rmv}}\emph{)} \in (\emph{log}_\emph{{recv}} \cup \emph{log}_\emph{{local}}) : \emph{id}_1 = \emph{id}_3 \land \emph{ts}_1.\emph{val} \leq \emph{vc}_\emph{{rmv}}[\emph{ts}_1.\emph{siteId}]\}$
		\State \tab $\emph{rmvs} = \{\emph{rmv(id}_1, \emph{vc}_1\emph{)} \in \emph{log}_\emph{{local}}:$ $\exists \emph{rmv(id}_2, \emph{vc}_2\emph{)} \in (\emph{log}_\emph{{local}} \cup \emph{log}_\emph{{recv}}) : \emph{id}_1 = \emph{id}_2 \wedge \emph{vc}_1 < \emph{vc}_2\}$
		\State \tab \textbf{return} $\emph{adds} \cup \emph{rmvs}$
		\\
		\State \function{mayHaveObservableImpact}$(\emph{log}_\emph{{local}}, S, \emph{log}_\emph{{recv}})$: set of operations
		\State \tab \textbf{return} $\{\}$ \Comment{This case never happens for this data type}
		\\
		\State \function{hasObservableImpact}$(\emph{log}_\emph{{local}}, S, \emph{log}_\emph{{recv}})$: set of operations
		\State \tab $\emph{adds} = \{\emph{add(id}_1, \emph{score}_1, \emph{ts}_1\emph{)} \in \emph{log}_\emph{{local}} : \langle \emph{id}_1, \emph{score}_1, \emph{ts}_1 \rangle \in \emph{topK(}S.\emph{elems)}\}$
		\State \tab $\emph{rmvs} = \{\emph{rmv(id}_1, \emph{vc}_1\emph{)} \in \emph{log}_\emph{{local}} : (\exists \emph{add(id}_2, \emph{score}_2, \emph{ts}_2\emph{)} \in (\emph{log}_\emph{{local}} \cup \emph{log}_\emph{{recv}}):$
		\State \tab\tab\tab $\langle \emph{id}_2, \emph{score}_2, \emph{ts}_2 \rangle \in \emph{topK(}S.\emph{elems}\cup \{\langle \emph{id}_2, \emph{score}_2, \emph{ts}_2 \rangle\}\emph{)} \wedge \emph{id}_1 = \emph{id}_2 \wedge \emph{ts}_2.\emph{val} \leq \emph{vc}_1[\emph{ts}_2.\emph{siteId}])\}$
		\State \tab \textbf{return} $\emph{adds} \cup \emph{rmvs}$
		\\
		\State \function{compact}$(ops)$: set of operations
		\State \tab \textbf{return} $ops$ \Comment{This data type does not require compaction}
	\end{algorithmic}
\end{algorithm}

Algorithm \ref{topk-with-deletes} presents a design that implements this semantics.
The prepare-update \emph{add} operation generates an effect-update \emph{add} that has an additional 
parameter consisting in a timestamp $(replica id, val)$, with $val$ a monotonically
increasing integer.
The prepare-update \emph{rmv} operation generates an effect-update \emph{rmv} that includes
an additional parameter consisting in a vector clock that summarizes add operations
that happened before the remove operation.
To this end, the object maintains a vector clock that is updated when a new add
is generated or executed locally. 
Additionally, this vector clock should be updated whenever a replica receives a message
from a remote replica (to summarize also the adds known in the sender that have not
been propagated to this replica).

Besides this vector clock, $vc$, each object replica maintains:
(i) a set, $elems$, with the elements added by all
	\emph{add} operations known locally (and that have not been removed yet); and
(ii) a map, $removes$, that maps each element $id$ to a vector clock with a summary of
	the add operations that happened before all removes of $id$ (for simplifying the 
	presentation of the algorithm, we assume that a key absent from the map has associated 
	a default vector clock consisting of zeros for every replica).

The execution of an \emph{add} consists in adding the element to the set
of $elems$ if the add has not happened before a previously received
remove for the same element -- this can happen as operations are not 
necessarily propagated in causal order.
The execution of a \emph{rmv} consists in updating $removes$ and deleting  
from $elems$ the information for adds of the element that happened before 
the remove.
To verify if an add has happened before a remove, we check if
the timestamp associated with the add is reflected in the remove vector clock
of the element (lines \ref{alg:topk:wasdeleted} and \ref{alg:topk:todelete}).
This ensures the intended semantics for the CRDT, assuming that the functions used
by the protocol are correct.

We now analyze the code of these functions.

Function \function{maskedForever} computes: the local adds that become
masked by other local adds (those for the same element with a lower
value) and removes (those for the same element that happened before the remove); 
the local removes that become masked by other removes (those for the same element that have
a smaller vector clock). 
In the latter case, it is immediate that a remove with a smaller vector clock becomes
irrelevant after executing the one with a larger vector clock.
In the former case, a local add for an element is masked by a more recent local add 
for the same element but with a larger value
as it is not possible to remove only the effects of the later add without removing the 
effect of the older one. 
A local add also becomes permanently masked by a local or remote remove that happened after 
the add.

Function \function{mayHaveObservableImpact} returns the empty set, as
for having impact on any observable state, an operation also has to have impact
on the local observable state by itself.

Function \function{hasObservableImpact} computes the local operations that 
are relevant for computing the top-K. An add is relevant if the added value is
in the top; a remove is relevant if it removes an add that would be otherwise in 
the top.

\subsection{Top Sum NuCRDT}
\label{sec:topsum}

We now present the design of a non-uniform CRDT, Top Sum, that maintains the 
top-K elements added to the object, where the value of each element is the sum 
of the values added for the element.
This data type can be used for maintaining a leaderboard in an online game where 
every time a player completes some challenge it is awarded some number of points, 
with the current score of the player being the sum of all points awarded.
It could also be used for maintaining a top of the best selling products 
in an (online) store (or the top customers, etc).

The semantics of the operations defined in the Top Sum object is the following.
The \emph{add(id, n)} update operation increments the value associated with \emph{id} by \emph{n}.
The \emph{get()} read-only operation returns the top-K mappings, $id \rightarrow value$,
as defined by the \emph{topK} function (similar to the Top-K NuCRDT).

This design is challenging, as it is hard to know which operations may have impact 
in the observable state. For example, consider a scenario with two replicas, where 
the value of the last element in the top is 100. 
If the known score of an element is 90, an add of 5 received in one replica may have 
impact in the observable state if the other replica has also received an add of 5 or more. 
One approach would be to propagate these operations, but this would lead to propagating 
all operations. 

To try to minimize the number of operations propagated we use the following heuristic 
inspired by the demarcation protocol and escrow transactions~\cite{demarcation,escrow}.
For each id that does not belong to the top, we compute the difference between the smallest 
value in the top and the value of the id computed by operations known in every replica -- this
is how much must be added to the id to make it to the top: let $d$ be this value.
If the sum of local adds for the id does not exceed $\frac{d}{num. replicas}$ in any replica, the value
of id when considering adds executed in all replicas is smaller than the smallest element in the top. 
Thus, it is not necessary to propagate add operations in this case, as they will not affect the top.

\begin{algorithm}[ht]
	\caption{Top Sum NuCRDT}\label{nu:topsum}
	\footnotesize	
	\begin{algorithmic}[1]
		\State $\emph{state}$ : map $\emph{id} \mapsto \emph{sum}$: initial $[]$
		\\
		\State \function{get}$()$ : map
		\State \tab \textbf{return} $\emph{topK(state)}$
		\\
		\State \textbf{prepare} \function{add}$(\emph{id}, \emph{n})$
		\State \tab \textbf{generate} \emph{add(id, n)}
		\\
		\State \textbf{effect} \function{add}$(\emph{id}, \emph{n})$
		\State \tab $\emph{state}[\emph{id}] = \emph{state}[\emph{id}] + \emph{n}$
		\\
		\State \function{maskedForever}$(\emph{log}_\emph{{local}}, S, \emph{log}_\emph{{recv}})$: set of operations
		\State \tab \textbf{return} $\{\}$ \Comment{This case never happens for this data type}
		\\
		\State \function{mayHaveObservableImpact}$(\emph{log}_\emph{{local}}, S, \emph{log}_\emph{{recv}})$: set of operations
		\State \tab $\emph{top} = \emph{topK(S.state)}$
		\State \tab $\emph{adds} = \{\emph{add(id, }\_\emph{)} \in \emph{log}_\emph{{local}} : \emph{s} = \emph{sum}_{\emph{val}}(\{\emph{add(i, n)} \in \emph{log}_\emph{{local}} : \emph{i} = \emph{id}\})$
		\State \tab\tab $\land$ $\emph{s} \geq ((\emph{min(sum(top))} - (S.state[id] - s))$ $/$ $\emph{getNumReplicas()})\}$
		\State \tab \textbf{return} \emph{adds}
		\\
		\State \function{hasObservableImpact}$(\emph{log}_\emph{{local}}, S, \emph{log}_\emph{{recv}})$: set of operations
		\State \tab $\emph{top} = \emph{topK(S.state)}$
		\State \tab $\emph{adds} = \{\emph{add(id, \_)} \in \emph{log}_\emph{{local}} : \emph{id} \in ids(top)\}$
		\State \tab \textbf{return} \emph{adds}
		\\
		\State \function{compact}$(\emph{ops})$: set of operations
		\State \tab $\emph{adds} = \{\emph{add(id, n)} : \emph{id} \in \{i : \emph{add(i, \_)} \in \emph{ops}\} \land$ $\emph{n} = \emph{sum}(\{\emph{k} : \emph{add(id}_1\emph{, k)} \in \emph{ops} : \emph{id}_1 = \emph{id}\})\}$
		\State \tab \textbf{return} \emph{adds}
	\end{algorithmic}
\end{algorithm}

Algorithm~\ref{nu:topsum} presents a design that implements this approach. 
The state of the object is a single variable, \emph{state}, that maps identifiers to their 
current values. 
The only prepare-update operation, \emph{add}, generates an effect-update \emph{add} 
with the same parameters. The execution of an effect-update \emph{add(id, n)} simply 
increments the value of \emph{id} by \emph{n}.

Function \function{maskedForever} returns the empty set, as operations in this design 
can never be forever masked. 

Function \function{mayHaveObservableImpact} computes the set of \emph{add} operations 
that can potentially have an impact on the observable state, using the approach previously explained.

Function \function{hasObservableImpact} computes the set of \emph{add} operations that 
have their corresponding \emph{id} present in the top-K. This guarantees that the values 
of the elements in the top are kept up-to-date, reflecting all executed operations.

Function \function{compact} takes a set of \emph{add} operations and compacts 
the \emph{add} operations that affect the same identifier into a single operation.
This reduces the size of the messages sent through the network and is similar to
the optimization obtained in delta-based CRDTs~\cite{delta}.

\subsection{Discussion}

The goal of non-uniform replication is to allow replicas to store less data and 
use less bandwidth for replica synchronization.
Although it is clear that non-uniform replication cannot be useful for all data, 
we believe that the number of use cases is large enough for making non-uniform 
replication interesting in practice. 
We now discuss two classes of data types that can benefit from the adoption of 
non-uniform replication.

The first class is that of data types for which the result of queries include only 
a subset of the data in the object.
In this case two different situations may occur: 
(i) it is possible to compute locally, without additional information, if some
operation is relevant (and needs to be propagated to all replicas);
(ii) it is necessary to have additional information to be able to decide if
some operation is relevant.
The Top-K CRDT presented in section~\ref{sec:topkrmv} is an example of the former. 
Another example includes a data type that returns a subset of the elements added
based on a (modifiable) user-defined filter -- e.g. in a set of books, the filter could 
select the books of a given genre, language, etc. 
The Top-Sum CRDT presented in section~\ref{sec:topsum} is an example of the latter.
Another example includes a data type that returns the $50^{th}$ percentile (or others)
for the elements added -- in this case, it is only necessary to replicate the elements 
in a range close to the $50^{th}$ percentile and replicate statistics of the elements
smaller and larger than the range of replicated elements.

In all these examples, the effects of an operation that in a given moment do not influence
the result of the available queries may become relevant after other operations are 
executed -- in the Top-K with removes 
due to a remove of an element in the top; in the filtered set due to a change in the filter;
in the Top-Sum due to a new add that makes an element relevant; and in the
percentile due to the insertion of elements that make the $50^{th}$ percentile change.
We note that if the relevance of an operation cannot change over time, the non-uniform CRDT
would be similar to an optimized CRDT that discard operations that are not relevant before 
propagating them to other replicas.

A second class is that of data types with queries that return the result of an aggregation 
over the data added to the object.
An example of this second class is the Histogram CRDT presented in the appendix.
This data type only needs to keep a count for each element. 
A possible use of this data type would be for maintaining the summary of classifications
given by users in an online shop.
Similar approaches could be implemented for data types that return the result
of other aggregation functions that can be incrementally 
computed \cite{Navalho:2015:SCC:2745947.2745948}.   

A data type that supports, besides adding some information, an operation for removing
that information would be more complex to implement. 
For example, in an Histogram CRDT that supports removing a previously added 
element, it would be necessary that concurrently removing the same element would not 
result in an incorrect aggregation result. 
Implementing such CRDT would require detecting and fixing these cases.

\section{Evaluation}
\label{sec:evaluations}

In this section we evaluate our data types that follow the non-uniform replication model. 
To this end, we compare our designs against state-of-the-art CRDT alternatives: 
delta-based CRDTs \cite{delta} that maintain full object replicas 
efficiently by propagating updates as deltas of the state; 
and computational CRDTs \cite{Navalho:2015:SCC:2745947.2745948} that maintain
non-uniform replicas using a state-based approach.

Our evaluation is performed by simulation, using a discrete event simulator. 
To show the benefit in terms of bandwidth and storage, we measure the total size of messages 
sent between replicas for synchronization (total payload) and the average size of replicas.

We simulate a system with 5 replicas for each object. Both our designs and the computational CRDTs 
support up to 2 replica faults by propagating all operations to, at least, 2 other replicas besides 
the source replica. We note that this limits the improvement that our approach could achieve,
as it is only possible to avoid sending an operation to two of the five replicas.
By either increasing the number of replicas or reducing the fault tolerance level, we could
expect that our approach would perform comparatively better than the delta-based CRDTs.

\subsection{Top-K with removals}

We begin by comparing our Top-K design (\emph{NuCRDT}) with a delta-based CRDT set~\cite{delta} (\emph{Delta CRDT}) 
and the top-K state-based computational CRDT design~\cite{Navalho:2015:SCC:2745947.2745948} (\emph{CCRDT}).

The top-K was configured with K equal to 100. In each run, 500000 update operations were generated for 10000 Ids and with scores up to 250000. The values used in each operation were randomly selected using a uniform distribution. A replica synchronizes after executing 100 events.

Given the expected usage of top-K for supporting a leaderboard, we expect the remove to be an infrequent operation (to be used only when a user is removed from the game). Figures~\ref{fig:topk1} and \ref{fig:topk3} show the results for workloads with $5\%$ and $0.05\%$ of removes respectively (the other operations are adds). 

\begin{figure*}[ht]
	\centering
	\includegraphics[width=0.44\linewidth]{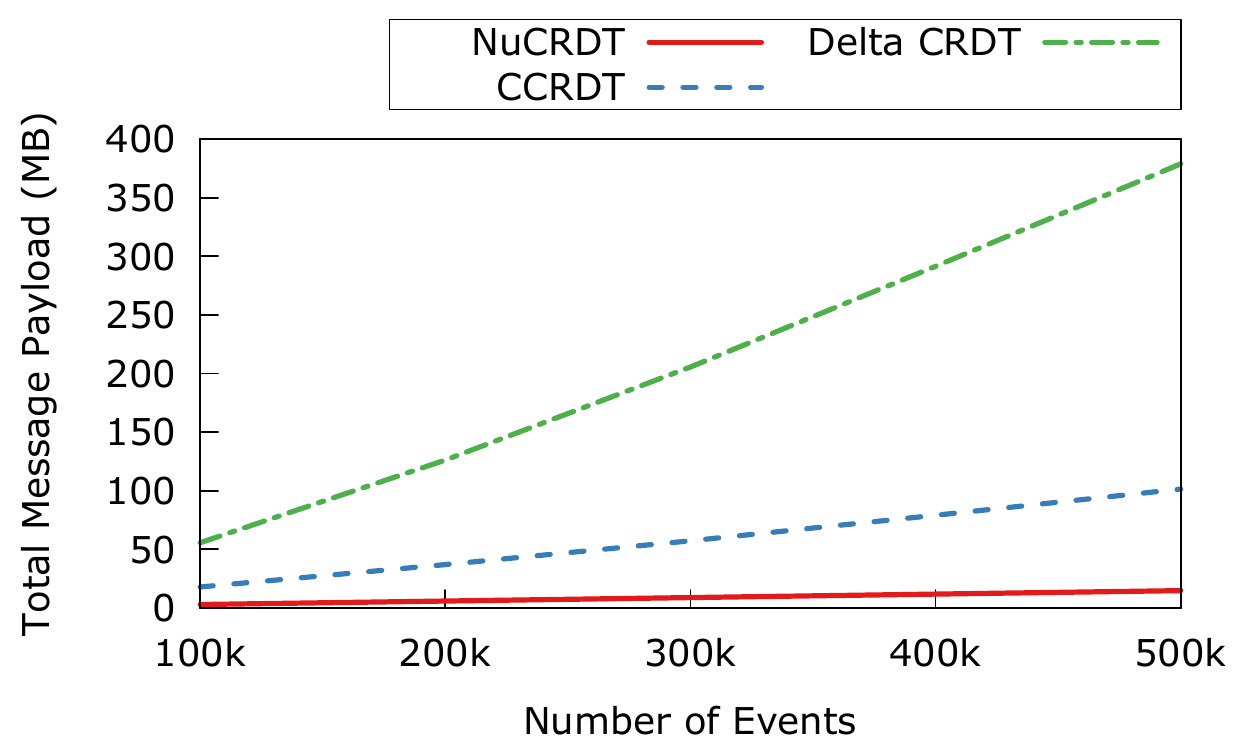}\label{fig:topk2}\hspace{0.08\linewidth}
	\includegraphics[width=0.44\linewidth]{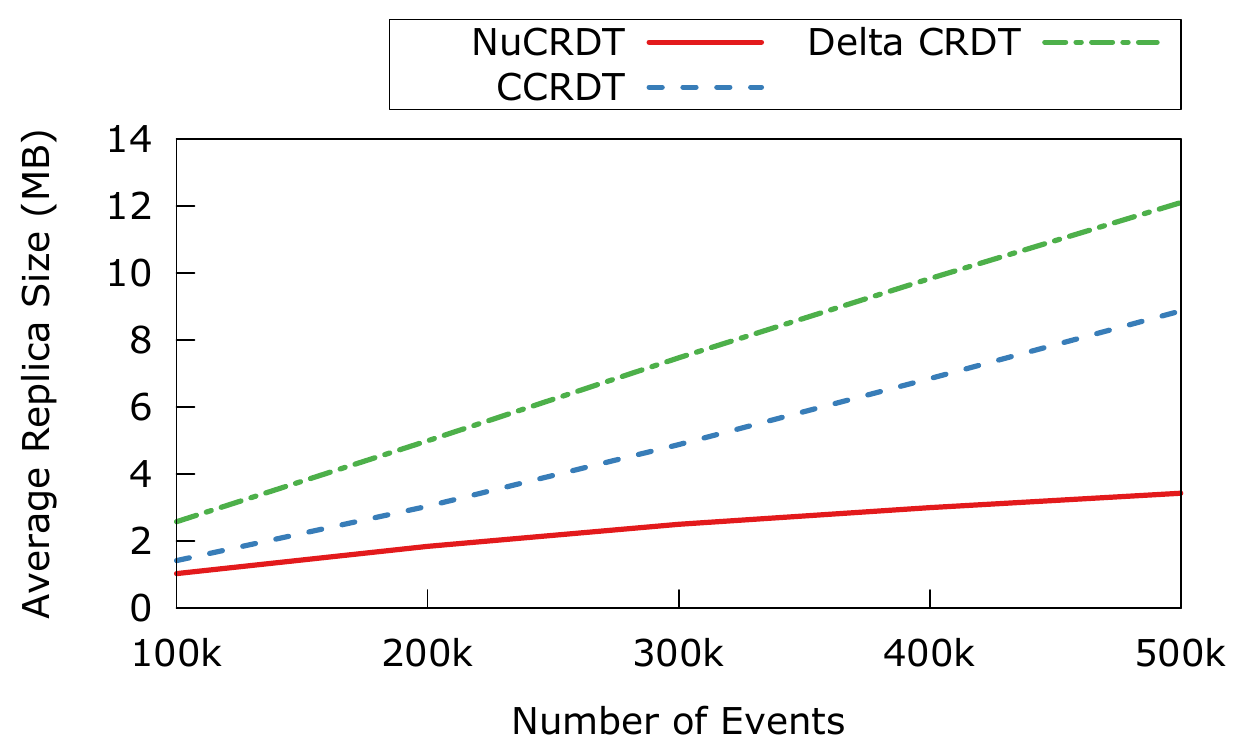}
	\caption{Top-K with removals: payload size and replica size, workload of 95/5}\label{fig:topk1}
\end{figure*}

\begin{figure*}[ht]
	\centering
	\includegraphics[width=0.45\linewidth]{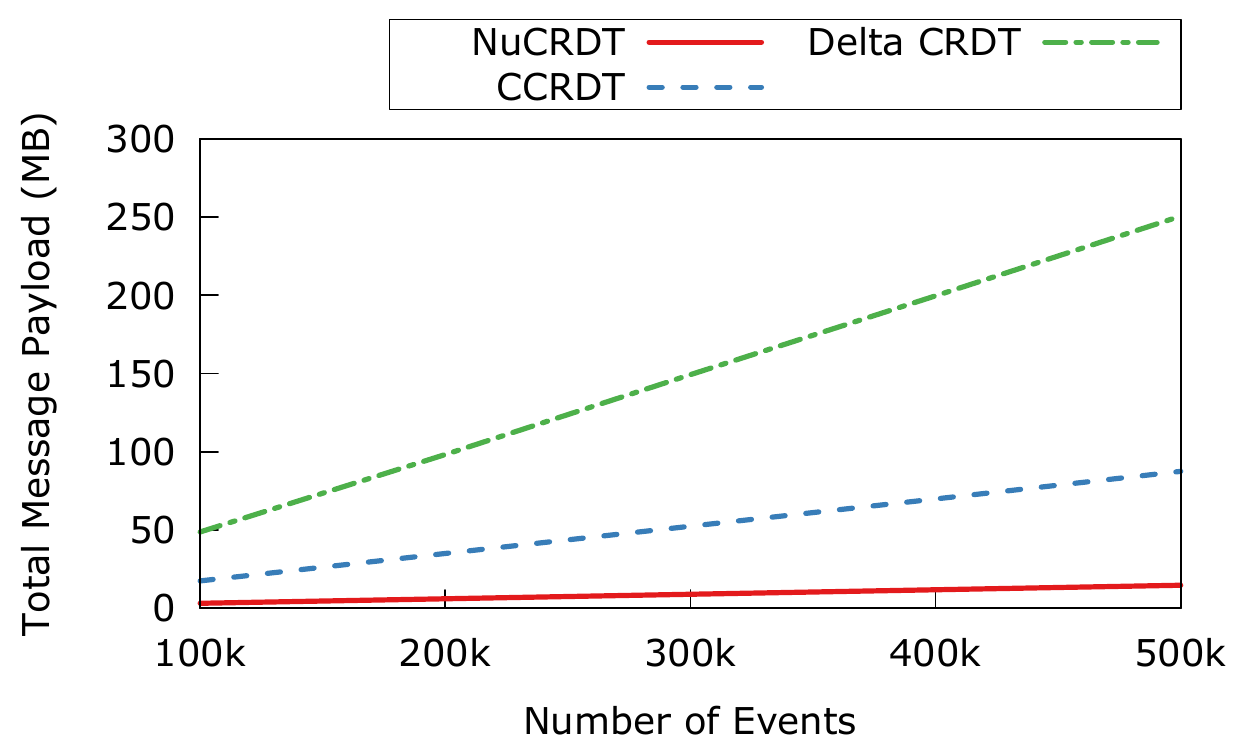}\hspace{0.08\linewidth}
	\includegraphics[width=0.45\linewidth]{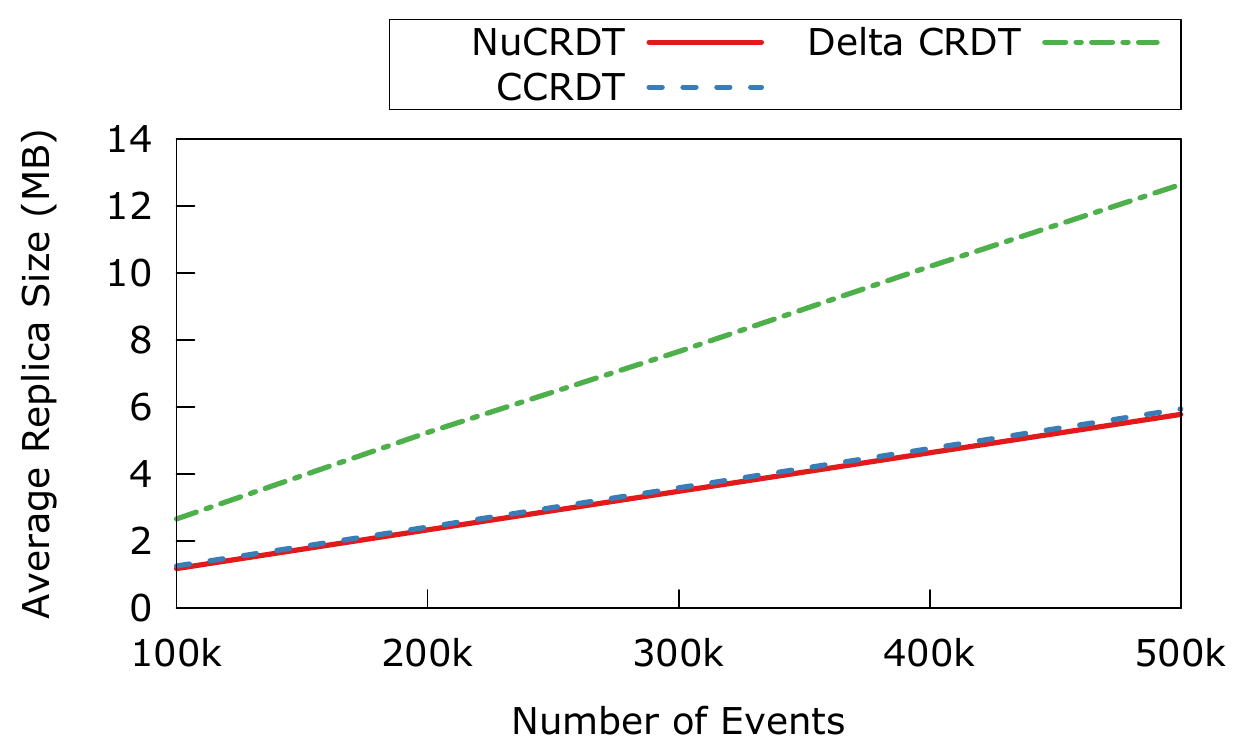}
	\caption{Top-K with removals: payload size and replica size, workload of 99.95/0.05}\label{fig:topk3}
\end{figure*}

In both workloads our design achieves a significantly lower bandwidth cost when compared to the alternatives. The reason for this is that our design only propagates operations that will be part of the top-K. In the delta-based CRDT, each replica propagates all new updates and not only those that are part of the top. In the computational CRDT design, every time the top is modified, the new top is propagated. Additionally, the proposed design of computational CRDTs always propagates removes.

The results for the replica size show that our design is also more space efficient than previous designs. This is a consequence of the fact that each replica, besides maintaining information about local operations, only keeps information from remote operations received for guaranteeing fault-tolerance and those that have influenced the top-K at some moment in the execution. The computational CRDT design additionally keeps information about all removes. The delta-based CRDT keeps information about all elements that have not been removed or overwritten by a larger value. We note that as the percentage of removes approaches zero, the replica sizes of our design and that of computational CRDT starts to converge to the same value. The reason for this is that the information maintained in both designs is similar and our more efficient handling of removes starts becoming irrelevant. The opposite is also true: as the number of removes increases, our design becomes even more space efficient when compared to the computational CRDT.

\subsection{Top Sum}

To evaluate our Top Sum design (\emph{NuCRDT}), we compare it against a delta-based CRDT map (\emph{Delta CRDT}) and a state-based computational CRDT implementing the same semantics (\emph{CCRDT}).

The top is configured to display a maximum of 100 entries. In each run, 500000 update operations were generated for 10000 Ids and with challenges awarding scores up to 1000. The values used in each operation were randomly selected using a uniform distribution. A replica synchronizes after executing 100 events.

\begin{figure*}[ht]
	\centering
	\includegraphics[width=0.44\linewidth]{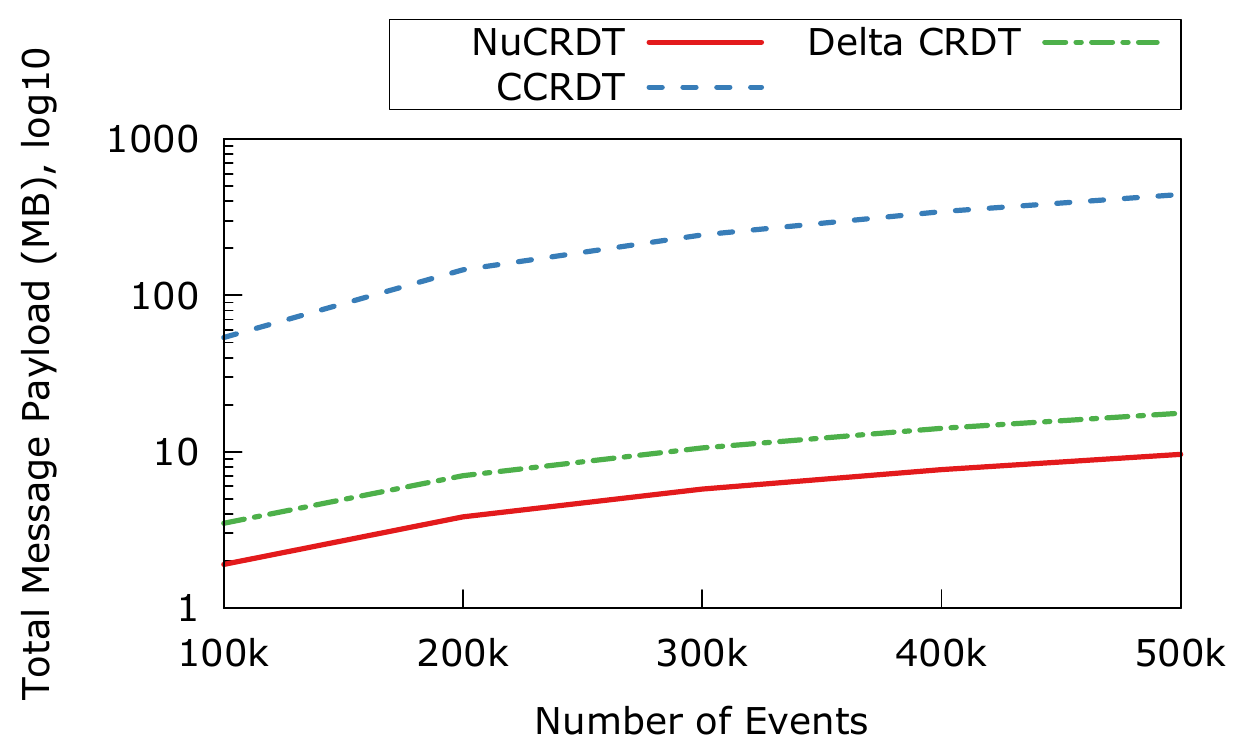}\hspace{0.08\linewidth}
	\includegraphics[width=0.44\linewidth]{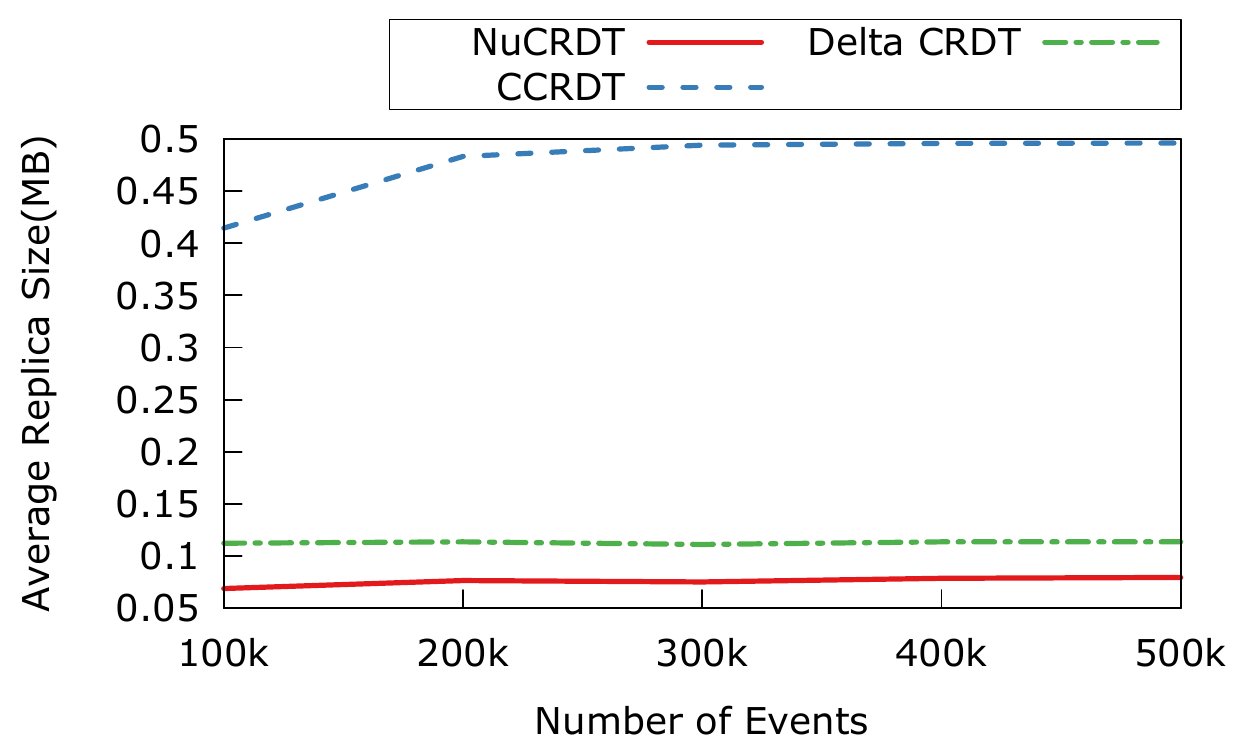}
	\caption{Top Sum: payload size and replica size}\label{fig:topsum}
\end{figure*}

Figure~\ref{fig:topsum} shows the results of our evaluation.
Our design achieves a significantly lower bandwidth cost when compared with the computational CRDT, because 
in the computational CRDT design, every time the top is modified, the new top is propagated.
When compared with the delta-based CRDTs, the bandwidth of NuCRDT is approximately $55\%$ of the bandwidth used by delta-based CRDTs. As delta-based CRDTs also include a mechanism for compacting propagated updates, the improvement
comes from the mechanisms for avoiding propagating operations that will not affect the top elements, resulting in less messages being sent.

The results for the replica size show that our design also manages to be more space efficient than previous designs. This is a consequence of the fact that each replica, besides maintaining information about local operations, only keeps information of remote operations received for guaranteeing fault-tolerance and those that have influenced the top elements at some moment in the execution.

\section{Conclusions}

In this paper we proposed the non-uniform replication model, an alternative model for replication 
that combines the advantages of both full replication, by allowing any replica to reply to
a query, and partial replication, by requiring that each replica keeps only part of the data.
We have shown how to apply this model to eventual consistency, and proposed a 
generic operation-based synchronization protocol for providing non-uniform replication. 
We further presented the designs of two useful replicated data types, the Top-K and Top Sum, that adopt this model (in appendix, we present two additional designs: Top-K without removals and Histogram).
Our evaluation shows that the application of this new replication model helps to reduce the message dissemination costs and the size of replicas.

In the future we plan to study which other data types can be designed that adopt this model
and to study how to integrate these data types in cloud-based databases.
We also want to study how the model can be applied to strongly consistent systems.

\section*{Acknowledgments}
This work has been partially funded by CMU-Portugal research project GoLocal Ref. CMUP-ERI/TIC/0046/2014, EU LightKone (grant agreement n.732505) and by FCT/MCT project NOVA-LINCS Ref. UID/CEC/04516/2013. Part of the computing resources used in this research were provided by a Microsoft Azure Research Award.

\bibliography{paper}

\newpage
\appendix

\section{APPENDIX}
In this appendix we present two additional NuCRDT designs. These designs exemplify the use of
different techniques for the creation of NuCRDTs.

\subsection{Top-K without removals}
\label{appendix-topk-normv}

A simpler example of a data type that fits our proposed replication model is a plain top-K, without support for the remove operation. This data type allows access to the top-K elements added to the object and can be used, for example, for maintaining a leaderboard in an online game. The top-K defines only one update operation, \emph{add(id,score)}, which adds element \emph{id} with score \emph{score}. The \emph{get()} operation simply returns the K elements with largest scores.  Since the data type does not support removals, and elements added to the top-K which do not fit will simply be discarded this means the only case where operations have an impact in the observable state are if they are core operations -- i.e. they are part of the top-K. This greatly simplifies the non-uniform replication model for the data type.

\begin{algorithm}[ht]
	\caption{Top-K NuCRDT}\label{topk-normv}
	\footnotesize
	\begin{algorithmic}[1]
		\State $elems : \{\langle id, score \rangle\}$ : initial $\{\}$
		\\
		\State \function{get}$()$ : set
		\State \tab \textbf{return} $elems$
		\\
		\State \textbf{prepare} \function{add}$(\emph{id}, \emph{score})$
		\State \tab \textbf{generate} $\emph{add(id}, \emph{score})$
		\\
		\State \textbf{effect} \function{add}$(\emph{id}, \emph{score})$
		\State \tab $elems = topK(elems \cup \{\langle id, score \rangle\})$
		\\
		\State \function{maskedForever}$(log_{local}, S, log_{recv})$ : set of operations
		\State \tab $adds = \{add(id_1, score_1) \in log_{local} : (\exists add(id_2, score_2) \in log_{recv} : id_1 = id_2 \land score_2 > score_1)$
		\State \tab \textbf{return} $adds$
		\\
		\State \function{mayHaveObservableImpact}$(log_{local}, S, log_{recv})$ : set of operations
		\State \tab \textbf{return} $\{\}$ \Comment{Not required for this data type}
		\\
		\State \function{hasObservableImpact}$(log_{local}, S, log_{recv})$ : set of operations
		\State \tab \textbf{return} $\{add(id, score) \in log_{local} : \langle id, score\rangle \in S.elems\}$
		\\
		\State \function{compact}$(ops)$: set of operations
		\State \tab \textbf{return} $ops$ \Comment{This data type does not use compaction}
	\end{algorithmic}
\end{algorithm}

Algorithm~\ref{topk-normv} presents the design of the top-K NuCRDT. The prepare-update \emph{add(id,score)} generates an effect-update \emph{add(id,score)}.

Each object replica maintains only a set of K tuples, \emph{elems}, with each tuple being composed of an \emph{id} and a \emph{score}. The execution of \emph{add(id,score)} inserts the element into the set, \emph{elems}, and computes the top-K of \emph{elems} using the function \emph{topK}. The order used for the \emph{topK} computation is as follows: $\langle \emph{id}_1, \emph{score}_1 \rangle > \langle \emph{id}_2, \emph{score}_2 \rangle$ iff $\emph{score}_1 > \emph{score}_2 \lor (\emph{score}_1 = \emph{score}_2 \land \emph{id}_1 > \emph{id}_2)$. We note that the \emph{topK} function returns only one tuple for each element \emph{id}.

Function \function{maskedForever} computes the adds that become masked by other add operations for the same \emph{id} that are larger according to the defined ordering. Due to the way the top is computed, the lower values for some given \emph{id} will never be part of the top. Function \function{mayHaveObservableImpact} always returns the empty set since operations in this data type are always core or forever masked. Function \function{hasObservableImpact} returns the set of unpropagated add operations which add elements that are part of the top -- essentially, the add operations that are core at the time of propagation. Function \function{compact} simply returns the given \emph{ops} since the design does not require compaction.

\subsection{Histogram}
\label{appendix-histogram}

We now introduce the Histogram NuCRDT that maintains a histogram of values added to the object. To this end, the data type maintains a mapping of bins to integers and can be used to maintain a voting system on a website. The semantics of the operations defined in the histogram is the following: \emph{add(n)} increments the bin \emph{n} by 1; \emph{merge(histogram}$_\emph{{delta}}\emph{)}$ adds the information of a histogram into the local histogram; \emph{get()} returns the current histogram.

\begin{algorithm}[ht]
	\caption{Histogram NuCRDT}\label{histogram}
	\footnotesize
	\begin{algorithmic}[1]
		\State $histogram :$ map $bin \mapsto n$ : initial $[]$
		\\
		\State \function{get}$()$ : map
		\State \tab \textbf{return} $histogram$
		\\
		\State \textbf{prepare} \function{add}$(\emph{bin})$
		\State \tab \textbf{generate} $\emph{merge}([\emph{bin} \mapsto 1])$
		\\
		\State \textbf{prepare} \function{merge}$(\emph{histogram})$
		\State \tab \textbf{generate} $\emph{merge}(\emph{histogram})$
		\\
		\State \textbf{effect} \function{merge}$(\emph{histogram}_{\emph{delta}})$
		\State \tab $\emph{histogram} = \emph{pointwiseSum}(\emph{histogram}, \emph{histogram}_{\emph{delta}})$
		\\
		\State \function{maskedForever}$(log_{local}, S, log_{recv})$ : set of operations
		\State \tab \textbf{return} $\{\}$ \Comment{Not required for this data type}
		\\
		\State \function{mayHaveObservableImpact}$(log_{local}, S, log_{recv})$ : set of operations
		\State \tab \textbf{return} $\{\}$ \Comment{Not required for this data type}
		\\
		\State \function{hasObservableImpact}$(log_{local}, S, log_{recv})$ : set of operations
		\State \tab \textbf{return} $log_{local}$
		\\
		\State \function{compact}$(ops)$: set of operations
		\State \tab $\emph{deltas} = \{\emph{hist} : \emph{merge}(\emph{hist}_{\emph{delta}}) \in ops\}$
		\State \tab \textbf{return} $\{\emph{merge}(\emph{pointwiseSum}(\emph{deltas}))\}$
	\end{algorithmic}
\end{algorithm}

This data type is implemented in the design presented in Algorithm~\ref{histogram}. The prepare-update \emph{add(n)} generates an effect-update \emph{merge(}[\emph{n} $\mapsto 1$]\emph{)}. The prepare-update operation \emph{merge(histogram)} generates an effect-update \emph{merge(histogram)}.

Each object replica maintains only a map, \emph{histogram}, which maps \emph{bins} to integers. The execution of a \emph{merge(histogram}$_\emph{{delta}}$\emph{)} consists of doing a pointwise sum of the local histogram with \emph{histogram}$_\emph{{delta}}$.

Functions \function{maskedForever} and \function{mayHaveObservableImpact} always return the empty set since operations in this data type are always core. Function \function{hasObservableImpact} simply returns $log_{local}$, as all operations are core in this data type. Function \function{compact} takes a set of instances of \emph{merge} operations and joins the histograms together returning a set containing only one \emph{merge} operation.

\end{document}